%% file: ms.tex
\documentclass[11pt]{article}

\usepackage[utf8]{inputenc}

\usepackage{amsmath,amsfonts,amsthm,amssymb,color}
\usepackage[usenames,dvipsnames,svgnames,table]{xcolor}
\definecolor{darkgreen}{rgb}{0.0,0,0.9}
\usepackage{mathtools}
\usepackage{authblk}
\usepackage{fullpage}
\usepackage{parskip}
\usepackage{comment}
\usepackage{tikz}
\usepackage{bbm}
\usepackage{dsfont}
\usepackage[sc]{mathpazo}
\usepackage[basic]{complexity}
\usepackage{algorithm2e}
\usepackage[colorlinks=true,
citecolor=OliveGreen,linkcolor=BrickRed,urlcolor=BrickRed,pdfstartview=FitH]{hyperref}
\usepackage[capitalize,nameinlink]{cleveref}
\usepackage{tcolorbox}
\usepackage[short]{optidef}

\newtcolorbox{wbox}
{
	colback  = white,
}

\SetKwInOut{Input}{Input}
\SetKwInOut{Output}{Output}
\SetKwFunction{Uncross}{\textsc{Uncross}}
\SetKwFunction{MergeUncross}{\textsc{MergeUncross}}
\SetKwFunction{PerfectMatching}{\textsc{PerfectMatching}}
\SetKwBlock{InParallel}{in parallel do}{end}
\SetKwFor{ParallelFor}{for}{in parallel do}{end}
\newcommand*{\Zplus}{\mathbb{Z_+}}

\newcommand*{\cR}{\mathcal{R}}

\newcommand*{\cA}{\mathcal{A}}
\newcommand*{\bE}{\mathbb{E}}
\newcommand*{\cO}{{{\rho}}}

\newtheorem{theorem}{Theorem}
\newtheorem{lemma}[theorem]{Lemma}
\newtheorem{corollary}[theorem]{Corollary}

\newtheorem{claim}[theorem]{Claim}

\theoremstyle{definition}
\newtheorem{definition}[theorem]{Definition}

\newtheorem{remark}[theorem]{Remark}

\newtheorem{alg}{Algorithm}

\newtheorem{example}[theorem]{Example}
\newtheorem{property}[theorem]{Property}

\newtheorem{ctheorem}[theorem]{Conditional Theorem}

\newenvironment{fminipage}%
{\begin{Sbox}\begin{minipage}}%
		{\end{minipage}\end{Sbox}\fbox{\TheSbox}}

\newcommand{\bid}{\mbox{\rm bid}}
\newcommand{\eb}{\mbox{\rm ebid}}

\newcommand{\ut}{{ut}}

\title{Towards a Practical, Budget-Oblivious Algorithm\\
for the Adwords Problem under Small Bids}

\author[1]{Vijay V.~Vazirani\footnote{Supported in part by NSF grants CCF-1815901 and CCF-2230414.}}

\affil[1]{University of California, Irvine}

\date{}

\begin{document}
	\maketitle
	
	\begin{abstract}	
Motivated by recent insights into the online bipartite matching problem (\textsc{OBM}), our goal was to extend  the optimal algorithm for it, namely \textsc{Ranking}, all the way to the special case of adwords problem, called \textsc{Small}, in which bids are small compared to budgets; the latter has been of considerable practical significance in ad auctions \cite{MSVV}. The attractive feature of our approach was that it would yield a {\em budget-oblivious algorithm}, i.e., the algorithm would not need to know budgets of advertisers and therefore could be used in autobidding platforms. 
	
We were successful in obtaining an optimal, budget-oblivious algorithm for \textsc{Single-Valued}, under which each advertiser can make bids of one value only. However, our next extension, to \textsc{Small}, failed because of a fundamental reason, namely failure of the {\em No-Surpassing Property}. Since the probabilistic ideas underlying our algorithm are quite substantial, we have stated them formally, after assuming the No-Surpassing Property, and we leave the open problem of removing this assumption. 

%Our algorithm also provides an avenue for improving the bound for the general adwords problem from $0.5016$, given in \cite{Huang2020adwords}. 

With the help of two undergrads, we conducted extensive experiments on our algorithm on randomly generated instances. Our findings are that the No-Surpassing Property fails less than $2\%$ of the time and that the performance of our algorithms for \textsc{Single-Valued} and \textsc{Small} are comparable to that of \cite{MSVV}. If further experiments confirm this, our algorithm may be useful as such in practice, especially because of its budget-obliviousness. 
\end{abstract}

\bigskip
\bigskip
\bigskip
\bigskip
\bigskip
\bigskip
\bigskip
\bigskip
\bigskip
\bigskip
\bigskip
\bigskip
\bigskip
\bigskip
\bigskip
\bigskip

\pagebreak

\input{Introduction}

\input{Prelim}

\input{Single-Valued}

\input{Adwords}

\input{discussion}

\input{ack}

	\bibliographystyle{alpha}
	\bibliography{refs}
	
%\appendix

%\input{Adwords}

\end{document}

%% file: Introduction.tex
\section{Introduction}
\label{sec.intro}

The {\em adwords problem}, called \textsc{Adwords}\footnote{For formal statements of problems studied in this paper, see Section \ref{sec.prelim}} in this paper, involves matching keyword queries, as they arrive online, to advertisers; the latter have daily budget limits and they make bids for the queries. Its special case when bids are small compared to budgets, called \textsc{Small} in this paper, captures a key computational issue that arises in the context of ad auctions, for instance in Google's AdWords marketplace. An optimal algorithm for \textsc{Small}, achieving a competitive ratio of $\left(1 - {1 \over e}\right)$, was first given in \cite{MSVV}; for the impact of this result in the marketplace, see Section \ref{sec.practical}. {\em In this paper, we give a new budget-oblivious online algorithm for \textsc{Small}}. 

A {\em budget-oblivious online algorithm} does not know the daily budgets of advertisers; however, in a run, it knows when the budget of an advertiser is exhausted. Yet its revenue is compared to the optimal revenue generated by an offline algorithm with full knowledge of the budget. The  importance of a budget-oblivious  algorithm lies in its use in autobidding platforms \cite{Auto-bidding, Chap-Adwords}, which manage the ad campaigns of large advertisers; they  dynamically adjust bids and budgets over multiple search engines to improve performance. In Open Problem Number 20, Mehta \cite{Mehta2013online} asks for such an algorithm for \textsc{Small}. 

Recent insights on the online bipartite matching problem (\textsc{OBM}) encouraged us to seek such an algorithm. A simple optimal algorithm, called \textsc{Ranking}, achieving a competitive ratio of $\left(1 - {1 \over e}\right)$, was given in \cite{KVV} for \textsc{OBM}. However, the analysis of \textsc{Ranking} given in \cite{KVV} was difficult to comprehend. A sequence of papers has finally led to a simple and elegant analysis, see Section \ref{sec.related}. The simplicity of \textsc{Ranking} is particularly attractive; moreover, it has become the paradigm-setting algorithmic idea in the area of online and matching-based market design \cite{Book-Online}. 

Ideas underlying the new proof of \textsc{OBM} enabled us to generalize \textsc{Ranking} all the way to an algorithm for \textsc{Small}, while retaining the simplicity of the \textsc{Ranking}. As a result of this simplicity, our algorithm has better properties than \cite{MSVV}; in particular, it is budget-oblivious. A detailed discussion of its running time is given below. A budget-oblivious algorithm for \textsc{Small}, having a competitive ratio  of $0.522$\footnote{Note that the greedy algorithm, which is clearly budget-oblivious, achieves a competitive ratio  of $0.5$.} was recently obtained by Udwani \cite{Unknown-udwani2021adwords}, using the idea of an LP-free analysis, which involves writing appropriate linear inequalities to compare the online algorithm with the offline optimal algorithm.  

At the outset of this work, extending \textsc{Ranking} directly to \textsc{Small} seemed an uphill task. Therefore we attempted an intermediate problem first, namely \textsc{Single-Valued}, in which each advertiser can make bids of one value only, although the value may be different for different advertisers. We note that \cite{Aggarwal2011online} had already obtained an optimal online algorithm for \textsc{Single-Valued} by reducing it to the vertex weighted online matching problem, see Section \ref{sec.related} for details. As explained in Section \ref{sec.ideas}, in order to develop tools for attacking \textsc{Small}, we needed to solve \textsc{Single-Valued} {\em directly}, and not resort to this reduction.

Our algorithm for \textsc{Single-Valued} is optimal, and it is also budget-oblivious. Furthermore, our algorithm uses fewer random bits than the approach of \cite{Aggarwal2011online}; see Section \ref{sec.related} for a detailed comparison. We note that  in contemporary\footnote{Our paper was first posted on arXiv on July 22, 2021 \cite{Va.OBM}.} and independent work, Albers and Schubert \cite{Albers2021optimal} obtained an identical result for \textsc{Single-Valued}; their technique is different and involves formulating a configuration LP and conducting a primal-dual analysis. Our technical ideas are described in Section \ref{sec.ideas}. 

Our analysis of \textsc{Single-Valued} involved new ideas from two domains, namely probability theory and combinatorics, with the former playing a dominant role and the latter yielding a  proof of a condition called the {\em No-Surpassing Property}, see Property \ref{prop.SV-no-surpass}. Equipped with these ideas, we next attempted an extension from  \textsc{Ranking} to \textsc{Small}. Although we met with success in extending the more difficult, probabilistic part, of the argument, we found a counter-example to the combinatorial part, showing that the No-Surpassing Property does not hold for \textsc{Small}. 

In order to make the no-surpassing property fail, we had to intricately ``doctored up'' the instance  of \textsc{Small}. This raised the question of experimentally determining how often this property fails in typical instances and how it affects the performance of our algorithm; for the latter, we compared it to \cite{MSVV}. As can be seen in the four Tables in Section \ref{sec.expt}, the property fails rarely, for less than $2\%$ of the edges $(i, j)$, and the performance of our algorithms for \textsc{Single-Valued} and \textsc{Small} are comparable with that of the MSVV Algorithm. For this reason, and because of its budget-obliviousness, the algorithm may be useful as such in practice. Clearly, it will be good to obtain further experimental confirm on varied  types of instances.

Since the ideas underlying our algorithm for \textsc{Small}, and the probabilistic part of its proof,   are quite substantial, we have stated them formally, after assuming the No-Surpassing Property, see Section \ref{sec.Adwords}. Under this assumption, we prove a competitive ratio of $\left(1 - {1 \over e}\right)$ for our algorithm. The problem of obtaining a tight unconditional competitive ratio of our algorithm is an important one and has received much attention over the last two years, ever since the appearance of this paper on arXiv. Critical insights into this open problem are provided by the following results: first, Udwani \cite{Unknown-udwani2021adwords} gave an example to show that the unconditional competitive ratio of our algorithm is strictly less than $(1 - 1/e)$. Next, Liang et al. \cite{Oblivious-0.624} showed that the unconditional competitive ratio is less than $0.624$; in contrast, $(1 - 1/e) \approx 0.632$.

\bigskip

\begin{remark}
	\label{rem.objective}
	The objective of all problems studied in this paper is to maximize the total revenue accrued by the online algorithm. In economics, such a solution is referred to as {\em efficient}, since the amount bid by an advertiser is indicative of how useful the query is to it, and hence to the economy.
\end{remark}

 \subsection{Related Works}
 \label{sec.related}

\textsc{OBM} occupies a central place not only in online algorithms but also in matching-based market design, see details in Section \ref{sec.practical}. The analysis of \textsc{Ranking} given in \cite{KVV} was considered ``difficult'' and it also had an  error. Over the years, several researchers contributed valuable ideas to simplifying its proof. The first simplifications, in \cite{Goel2008online, Claire2008line}, got the ball rolling, setting the stage for the substantial simplification given in \cite{Devanur2013randomized}, using a randomized primal-dual approach. \cite{Devanur2013randomized} introduced the idea of splitting the contribution of each matched edge into primal and dual contributions and lower-bounding each part separately. Their method for defining prices $p_j$ of goods, using randomization, was used by subsequent papers, including this one\footnote{For a succinct proof of optimality of the underlying function, $e^{x-1}$, see Section 2.1.1 in \cite{Chap-Online}.}. 

Interestingly enough, the next simplification involved removing the scaffolding of LP-duality and casting the proof in purely probabilistic terms\footnote{Even though there is no overt use of LP-duality in the proof of \cite{Eden2018economic}, it is unclear if this proof could have been obtained directly, without going the LP-duality-route.}, using notions from economics to split the contribution of each matched edge into the contributions of the buyer and the seller. This elegant analysis was given by \cite{Eden2018economic}. We note that when we move to generalizations of \textsc{OBM}, even this economic interpretation needs to be dropped, see Remark \ref{rem.names}. Building on these works, and incorporating a further simplification relating to the No-Surpassing Property for \textsc{OBM}, a ``textbook quality'' proof was recently given in \cite{Va.MFCS}.   

An important generalization of \textsc{OBM} is online $b$-matching. This problem is a special case of \textsc{Adwords} in which the budget of each advertiser is $\$b$ and the bids are $0/1$. \cite{Pruhs} gave a simple optimal online algorithm, called \textsc{BALANCE}, for this problem. \textsc{BALANCE} awards the next query to the interested bidder who has been matched least number of times so far. \cite{Pruhs} showed that as $b$ tends to infinity, the competitive ratio of \textsc{BALANCE} tends to $\left(1 - {1 \over e}\right)$.

Observe that $b$-matching is a special case of \textsc{Small}, if $b$ is large. Indeed, MSVV Algorithm was obtained by extending \textsc{BALANCE}\footnote{It is worth recalling that \cite{MSVV}  had first attempted extending \textsc{OBM} to \textsc{Small}; however, in the absence of new insights into \textsc{OBM}, this did not go very far.} as follows: \cite{MSVV} first gave a simpler proof of the competitive ratio of \textsc{BALANCE} using the notion of a {\em factor-revealing LP} \cite{Factor2003}. Then they gave the notion of a {\em tradeoff-revealing LP}, which yielded an algorithm achieving a competitive ratio of $\left(1 - {1 \over e}\right)$. \cite{MSVV} also proved that this is optimal for $b$-matching, and hence \textsc{Small}, by proving that no randomized algorithm can achieve a better ratio for online $b$-matching; previously, \cite{Pruhs} had shown a similar result for deterministic algorithms.

The MSVV Algorithm is simple and operates as follows. The effective bid of each bidder $j$ for a query is its bid multiplied by $(1 - e^{L_j/B_j})$, where $B_j$ and $L_j$ are the total budget and the leftover budget of bidder $j$, respectively; the query is matched to the bidder whose effective bid is highest. As a result, the MSVV Algorithm needs to know the total budget of each bidder. Following \cite{MSVV}, a second optimal online algorithm for \textsc{Small} was given in \cite{Buchbinder2007online}, using a primal-dual approach.  

Another relevant generalization of \textsc{OBM} is online vertex weighted matching, in which the offline vertices have weights and the objective is to maximize the weight of the matched vertices. \cite{Aggarwal2011online} extended \textsc{Ranking} to obtain an optimal online algorithm for this problem. Clearly, \textsc{Single-Valued} is intermediate between \textsc{Adwords} and online vertex weighted matching. \cite{Aggarwal2011online} gave an optimal online algorithm for \textsc{Single-Valued} by reducing it to online vertex weighted matching. This involved creating $k_j$ copies of each advertiser $j$. As a result, their algorithm needs to use $\sum_{j \in A} {k_j}$ random numbers, where $A$ is the set of advertisers. On the other hand, our algorithm, and that of \cite{Albers2021optimal}, needs to use only $|A|$ numbers.

\textsc{Adwords} is a notoriously difficult problem, partly due to  its inherent structural difficulties, which are described in Section \ref{sec.diff-GENERAL}. For \textsc{Adwords}, the greedy algorithm, which matches each query to the highest bidder, achieves a competitive ratio of $1/2$. Until recently, that was the best possible. In \cite{Huang2020adwords} a marginally improved algorithm, with a ratio of $0.5016$, was given. It is important to point out that this 60-page paper was a tour-de-force, drawing on a diverse collection of ideas --- a testament to the difficulty of this problem.

 In the decade following the conference version (FOCS 2005) of \cite{MSVV}, search engine companies generously invested in research on models derived from \textsc{OBM} and adwords. The reason  was two-fold: the substantial impact of \cite{MSVV} and the emergence of a rich collection of digital ad tools. It will be impossible to do justice to this substantial body of work, involving both algorithmic and game-theoretic ideas; for a start, see the surveys \cite{Mehta2013online, Chap-Online}.

\subsection{Significance and Practical Impact}
 \label{sec.practical}

Google's AdWords marketplace generates multi-billion dollar revenues annually and the current annual worldwide spending on digital advertising is almost half a trillion dollars. These revenues of Google and other Internet services companies enable them to offer crucial services, such as search, email, videos, news, apps, maps etc. for free -- services that have virtually transformed our lives. 

We note that \textsc{Small} is the most relevant case of adwords for the search ads marketplace e.g., see \cite{Chap-Adwords}. A remarkable feature of Google, and other search engines, is the speed with which they are able to show search results, often in milliseconds. In order to show ads at  the same speed, together with search results, the solution for \textsc{Small} needed to be minimalistic in its use of computing power, memory and communication.  

The MSVV Algorithm satisfied these criteria and therefore had substantial impact in this marketplace. Furthermore, the idea underlying their algorithm was extracted into a simple heuristic, called  {\em bid scaling}, which uses even less computation and is widely used by search engine companies today. As mentioned above, our Conditional Algorithm for \textsc{Small} is even more elementary and is budget-oblivious.

It will be useful to view the AdWords marketplace in the context of a bigger revolution, namely the advent of the Internet and mobile computing, and the consequent resurgence of the area of matching-based market design. The birth of this area goes back to the seminal 1962 paper of Gale and Shapley on stable matching \cite{GaleS}. Over the decades, this area became known for its highly successful applications, having economic as well as sociological impact. These included matching medical interns to hospitals, students to schools in large cities, and kidney exchange. 
 
  The resurgence led to a host of highly innovative and impactful applications. Besides the AdWords marketplace, which matches queries to advertisers, these include Uber, matching drivers to riders; Upwork, matching employers to workers; and Tinder, matching people to each other, see \cite{Simons,Book-Online} for more details. 
  
  A successful launch of such markets calls for economic and game-theoretic insights, together with algorithmic ideas. The Gale-Shapley Deferred Acceptance Algorithm and its follow-up works provided the algorithmic backbone for the ``first life'' of matching-based market design. The algorithm \textsc{Ranking} has become the paradigm-setting algorithmic idea in the ``second life'' of this area \cite{Book-Online}. Interestingly enough, this result was obtained in the pre-Internet days, over thirty years ago.

 \subsection{Technical Ideas}
 \label{sec.ideas}

Our extension from \textsc{Ranking} to \textsc{Small} needs to go via  \textsc{Adwords}. It turns out that \textsc{Adwords} suffers from an inherent structural difficulty, see Section \ref{sec.diff-GENERAL}. We temporarily finesse this difficulty by using the idea of ``fake'' money. The expected revenue of our online algorithm for \textsc{Adwords} is at least $(1 - 1/e)$ fraction of the optimal offline revenue; however, this total revenue consists of real as well as fake money. We provide an upper-bound on the fake money in the worst case, and this suffices to show that, asymptotically, the fake money used by  \textsc{Small}, is negligible. Determining the true competitive ratio of our algorithm for \textsc{Adwords} is left as an interesting and important open problem, see Section \ref{sec.discussion}. 

 As described in Section \ref{sec.related}, \textsc{Single-Valued} can be reduced to online vertex weighted matching, by making $k_j$ copies of each advertiser $j$; however, this reduction does not work for \textsc{Adwords}. The reason is that the manner in which budget $B_j$ of bidder $j$ gets partitioned into bids is not predictable in the latter problem; it depends on the queries, their order of arrival and the randomization executed in a run of the algorithm. Therefore, in order to build techniques to attack \textsc{Adwords}, we will first need to solve \textsc{Single-Valued} {\em without} reducing it to online vertex weighted matching. 

This is done in Algorithm \ref{alg.main}. Almost all of our new ideas, on the probabilistic front, needed to attack \textsc{Small} were obtained in the process analyzing this algorithm. First, since vertex $j$ is not split into $k_j$ copies, we cannot talk about the contribution of edges anymore. Even worse, we don't have individual vertices for keeping track of the revenue accrued from each match, as per the scheme of \cite{Eden2018economic}. Our algorithm gets around this difficulty by accumulating revenue in the same ``account'' each time bidder $j$ gets matched. The corresponding random variable, $r_j$, is called the {\em total revenue} of bidder $j$, for want of a better name, see Remark \ref{rem.names}. Lower bounding $\bE[r_j]$ is much more tricky than lower bounding the revenue of a good in \textsc{OBM}, since it involves ``teasing apart'' the $k_j$ accumulations made into this account; this is done in Lemma \ref{lem.star}. 

 The key fact needed in the analysis of \textsc{Ranking} is that for each edge $e = (i, j)$ in the underlying graph, its expected contribution to the matching produced is at least $(1 - 1/e)$. For this purpose, the random variable, $u_e$, called {\em threshold}, is defined in \cite{Va.MFCS}.
 
 For analyzing \textsc{Single-Valued}, a replacement is needed for this lemma. For this purpose, we give the notion of a {\em $j$-star}, denoted $X_j$, which consists of bidder $j$ together with edges to $k_j$ of its neighbors in $G$, see Definition \ref{def.j-star}. The contribution of $j$-star $X_j$, is denoted by $\bE[X_j]$, which is also defined in Definition \ref{def.j-star}. Finally, using the lower bound on $\bE[r_j]$, Lemma \ref{lem.star} gives a lower $\bE[X_j]$ for {\em every} $j$-star, $X_j$. This lemma crucially uses a new random variable, called {\em truncated threshold}, see Definition \ref{def.threshold-Single}. 

Next, we explain the reason for truncation in the definition of this random variable. Consider bidder $j$ and a query $i_l$ that is desired by $j$. Observe that in run $\cR_j$\footnote{Run $\cR_j$ is defined in Definition \ref{def.R_j}.}, query $i_l$ can get a bid as large as $B \cdot (1 - {1 \over e})$, where $B = \max_{k \in A} \{b_k\}$, whereas the largest bid that $j$ can make to $i_l$ is $b_j  \cdot (1 - {1 \over e})$; in general, $b_j$ may be smaller than $B$. Now, $i_l$ contributes revenue to $r_j$ only if $i_l$ is matched to $j$ in run $\cR$, an event which will definitely not happen if $u_{e_l} > b_j  \cdot (1 - {1 \over e})$. Therefore, whenever $u_{e_l} \in [b_j \cdot (1 - {1 \over e}), \ B \cdot (1 - {1 \over e})]$, the contribution to $r_j$ is zero. By truncating $u_{e_l}$ to $b_j \cdot (1 - {1 \over e})$, we have effectively changed the probability density function of $u_{e_l}$ so that the probability of the event $u_{e_l} \in [b_j \cdot (1 - {1 \over e}), B \cdot (1 - {1 \over e})]$ is now concentrated at the event $u_{e_l} = b_j \cdot (1 - {1 \over e})$. From the viewpoint of lower bounding the revenue accrued in $r_j$, the two probability density functions are equivalent since the revenue accrued is zero under both these events. On the other hand, the truncated random variable enables us to apply the law of total expectation, in the proof of Lemma \ref{lem.star}, in the same way as it was done in \cite{Va.MFCS}, without introducing more difficulties.   	  

Finally, in order to establish the no-surpassing property for \textsc{Single-Valued}, we give the necessary combinatorial facts in Lemma \ref{lem.S-subset-j} and Corollary \ref{cor.S-subset-j}. These facts are enhanced versions of the facts needed to prove the no-surpassing property for \textsc{Ranking} in \cite{Va.MFCS}.

%% file: Prelim.tex
\section{Preliminaries}
\label{sec.prelim}

{\bf Online Bipartite Matching} (\textsc{OBM}):
Let $B$ be a set of $n$ buyers and $S$ a set of $n$ goods. A bipartite graph $G = (B, S, E)$ is specified on vertex sets $B$ and $S$, and edge set $E$, where for $i \in B, \ j \in S$, $(i, j) \in E$ if and only if buyer $i$ {\em likes} good $j$. $G$ is assumed to have a perfect matching and therefore each buyer can be given a unique good she likes. Graph $G$ is revealed in the following manner. The $n$ goods are known up-front. On the other hand, the buyers arrive one at a time, and when buyer $i$ arrives, the edges incident at $i$ are revealed.  

We are required to design an online algorithm $\cA$ in the following sense. At the moment a buyer $i$ arrives, the algorithm needs to match $i$ to one of its unmatched neighbors, if any; if all of $i$'s neighbors are matched, $i$ remains unmatched. The difficulty is that the algorithm does not ``know'' the edges incident at buyers which will arrive in the future and yet the size of the matching produced by the algorithm will be compared to the best {\em off-line matching}; the latter of course is a perfect matching. The formal measure for the algorithm is defined in Section \ref{sec.competitive}.

{\bf Adwords Problem} (\textsc{Adwords}):
Let $A$ be a set of $m$ {\em advertisers}, also called {\em bidders}, and $Q$ be a set of $n$ {\em queries}. A bipartite graph $G = (Q, A, E)$ is specified on vertex sets $Q$ and $A$, and edge set $E$, where for $i \in Q$ and $j \in A$, $(i, j) \in E$ if and only if bidder $j$ is {\em interested in} query $i$. Each query $i$ needs to be matched\footnote{Clearly, this is not a matching in the usual sense, since a bidder may be matched to several queries.} to at most one bidder who is interested in it. For each edge $(i, j)$, bidder $j$ {\em knows} his bid for $i$, denoted by $\bid(i, j) \in \Zplus$. Each bidder also has a {\em budget} $B_j \in \Zplus$ which satisfies $B_j \geq \bid(i, j)$, for each edge $(i, j)$ incident at $j$. 

Graph $G$ is revealed in the following manner. The $m$ bidders are known up-front and the queries  arrive one at a time. When query $i$ arrives, the edges incident at $i$ are revealed, together with the bids associated with these edges. If $i$ gets matched to $j$, then the matched edge $(i, j)$ is assigned a weight of $\bid(i, j)$. The constraint on $j$ is that the total weight of matched edges incident at it be at most $B_j$. The objective is to maximize the total weight of all matched edges at all bidders.

{\bf Adwords under Single-Valued Bidders} (\textsc{Single-Valued}):
\textsc{Single-Valued} is a special case of \textsc{Adwords} in which each bidder $j$ will make bids of a single value, $b_j \in \Zplus$, for the queries he is interested in. If $i$ accepts $j$'s bid, then $i$ will be matched to $j$ and the weight of this matched edge will be $b_j$. Corresponding to each bidder $j$, we are also given $k_j \in \Zplus$, the maximum number of times $j$ can be matched to queries. The objective is to maximize the total weight of matched edges. Observe that the matching $M$ found in $G$ is a $b$-matching with the $b$-value of each query $i$ being 1 and of advertiser $j$ being $k_j$.

{\bf Adwords under Small Bids} (\textsc{Small}):
\textsc{Small} is a special case of \textsc{Adwords} in which for each bidder $j$, each bid of $j$ is small compared to its budget. Formally, we will capture this condition by imposing the following constraint. For a valid instance $I$ of \textsc{Small}, define
\[ \mu(I) = \max_{j \in A} \ \ \left\{ {{\max_{(i, j) \in E} \ \{\bid(i, j) -1 \}} \over {B_j}} \right\} .\]
Then we require that
\[ \lim_{n(I) \rightarrow \infty} \ {\mu(I)} = 0 , \]
where $n(I)$ denotes the number of queries in instance $I$.

\begin{comment}
	
{\bf Edge-Weighted Bipartite Matching Problem:} 
Let $A$ be a set of $m$ {\em advertisers}, also called {\em bidders}, and $Q$ be a set of $n$ {\em impressions}. A bipartite graph $G = (Q, A, E)$ is specified on vertex sets $Q$ and $A$, and edge set $E$, where for $i \in Q$ and $j \in A$, $(i, j) \in E$ if and only if bidder $j$ is {\em interested in} impression $i$. Each impression $i$ needs to be matched to at most one bidder who is interested in it and at the end of the algorithm, each bidder needs to be allocated at most one impression. If $(i, j) \in E$, bidder $j$ has a {\em value}, $\bid(i, j) \in \Zplus$, for  impression $i$. The impressions arrive online, one at a time. When impression $i$ arrives, all edges incident at it are revealed, together with their values. The algorithm needs to match $i$ irrevocably to a bidder. The objective is to maximize the sum of values of matched edges.   
\end{comment}

\subsection{The competitive ratio of online algorithms}
\label{sec.competitive}

We will define the notion of competitive ratio of a randomized online algorithm in the context of \textsc{OBM}.

\begin{definition}
	\label{def.ratio}
	Let $G = (B, S, E)$ be a bipartite graph as specified above.
	The {competitive ratio} of a randomized algorithm $\cA$ for \textsc{OBM} is defined to be:
		\[  c(\cA) =   \min_{G = (B, S, E)} \min_{\cO(B)} \ {\frac {\bE[\cA(G, \cO(B))]} {n}} , \]
where $\bE[\cA(G, \cO(B))]$ is the expected size of matching produced by $\cA$; the expectation is over the random bits used by $\cA$. We may assume that the worst case graph and the order of arrival of buyers, given by $\cO(B)$, are chosen by an adversary who knows the algorithm. It is important to note that the algorithm is provided random bits {\em after} the adversary makes its choices. 
\end{definition}

\bigskip

\begin{remark}
	\label{rem.complete-matching}
	For each problem studied in this paper, we will assume that the offline matching is complete. It is easy to extend the arguments, without changing the competitive ratio, in case the offline matching is not complete. 
%As an example, this is done for \textsc{OBM} in Remark \ref{rem.not-perfect}.  
\end{remark}

%% file: Single-Valued.tex
\section{Algorithm for \textsc{Single-Valued}}
\label{sec.Single-Adwords}

Algorithm \ref{alg.main}, which will be denoted by $\cA_1$, is an online algorithm for \textsc{Single-Valued}. Before execution of Step (1) of $\cA_1$, the order of arrival of queries, say $\cO(B)$, is fixed by the adversary. We will define several random variables whose purpose will be quite similar to that in \textsc{Ranking} and they will be given similar names as well; however, their function is not as closely tied to these economics-motivated names as in \textsc{Ranking}, see also Remark \ref{rem.names}. Three of these random variables are the {\em price} $p_j$ and {\em total revenue} $r_j$ of each bidder $j \in A$, and the {\em utility} $u_i$ of each query $i \in Q$. 

We now describe how values are assigned to these random variables in a run of Algorithm \ref{alg.main}. In Step (1), for each bidder $j$, $\cA_1$ picks a price $p_j \in [{1 \over e}, 1]$ via the specified randomized process. Furthermore, the revenue $r_j$ and {\em degree} $d_j$ of bidder $j$ are both initialized to zero, the latter represents the number of times $j$ has been matched. During the run of $\cA_1$, $j$ will get matched to at most $k_j$ queries; each match will add $b_j$ to the total revenue generated by the algorithm. $b_j$ is broken into a revenue and a utility component, with the former being added to $r_j$ and the latter forming $u_i$. At the end of $\cA_1$, $r_j$ will contain all the revenue accrued by $j$.

In Step (2), on the arrival of query $i$, we will say that bidder $j$ is {\em available} if  $(i, j) \in E$ and $d_j < k_j$. At this point, for each available bidder $j$, the {\em effective bid} of $j$ for $i$ is defined to be $\eb(j) = b_j \cdot (1 - p_j)$; clearly, $\eb(j) \in [0, b_j \cdot \left(1 - {1 \over e}\right)]$. Query $i$ accepts the bidder whose effective bid is the largest. If there are no bids, matching $M$ remains unchanged. If $i$ accepts $j$'s bid, then edge $(i, j)$ is added to matching $M$ and the weight of this edge is set to $b_j$. Furthermore, the {\em utility} of $i$, $u_i$, is defined to be $\eb(j)$ and the revenue $r_j$ of $j$ is incremented  by $b_j \cdot p_j$. Once all queries are processed, matching $M$ and its weight $W$ are output.

\bigskip
		
\begin{remark}
\label{rem.names}
\cite{Eden2018economic} had given the economics-based names of random variables for their proof of \textsc{Ranking}. Although we have used the same names for similar random variables in Sections \ref{sec.Single-Adwords} and \ref{sec.Adwords}, for \textsc{Single-Valued} and \textsc{Adwords}, the reader should not attribute an economic interpretation to these the names\footnote{We failed to come up with more meaningful names for these random variables and therefore have stuck to the old names.}.
\end{remark}

\subsection{Analysis of Algorithm \ref{alg.main}}
\label{sec.analysis}

For the analysis of Algorithm $\cA_1$, we will use the random variables $W$, $p_j, r_j$ and $u_i$ defined above; their values are fixed during the execution of $ \cA_2$. In addition, corresponding to each edge $e = (i, j) \in E$, in Definition \ref{def.threshold-Single}, we will introduce a new random variable, $u_e$, which will play a central role. 

\bigskip

\begin{lemma}
	\label{lem.add-SINGLE}
	\[ \bE [W] \ =  \sum_i^n {\bE \left[ {u_i} \right]} \ + \   \sum_j^m {\bE[r_j]}  .\]
\end{lemma}

\begin{proof}
	For each edge $(i, j) \in M$, its contribution to $W$ is $b_j$. Furthermore, the sum of $u_i$ and the contribution of $(i, j)$ to $r_j$ is also $b_j$. This gives the first equality below. The second equality follows from linearity of expectation.
		\[ \bE [W] \ = \bE \left[ \sum_{i = 1}^n {u_i} \ + \ \sum_{j = 1}^m {r_j} \right] 
		\ =  \sum_i^n {\bE \left[ {u_i} \right]} \ + \   \sum_j^m {\bE[r_j]}   ,\]
\end{proof}

\bigskip

\setcounter{figure}{1} 
%\begin{figure}[H]

\begin{figure}

	\begin{wbox}
		\begin{alg}
		\label{alg.main}
		{\bf ($\cA_1$: Algorithm for \textsc{Single-Valued})}\\

		\begin{enumerate}
			\item {\bf Initialization:} $M \leftarrow \emptyset$. \\
			 $\forall j \in A$, do:
				\begin{enumerate}
				\item Pick $w_j$ uniformly from $[0, 1]$ and 
			set price \ $p_j \leftarrow e^{w_j - 1}$.
			\item $r_j \leftarrow 0$.
			\item $d_j \leftarrow 0$.
			\end{enumerate}
			
		\bigskip
		
		\item {\bf Query arrival:} When query $i$ arrives, {\bf do}:
			\begin{enumerate}
			\item  $\forall j \in A$ s.t. $(i, j) \in E$ and $d_j < k_j$ {\bf do}: 
			\begin{enumerate}
				\item $\eb(j) \leftarrow b_j \cdot (1 - p_j)$.
				\item Offer effective bid of $\eb(j)$ to $i$.
			\end{enumerate}

			\item Query $i$ accepts the bidder whose effective bid is the largest. \\
			      (If there are no bids, matching $M$ remains unchanged.) \\
				If $i$ accepts $j$'s bid, then {\bf do}:
			\begin{enumerate}
			\item Set utility: \ $u_i \leftarrow b_j \cdot (1 - p_j)$.
			\item Update revenue: \ $r_j \leftarrow r_j + b_j \cdot p_j$.
			\item Update degree: \ $d_j \leftarrow d_j + 1$.
			\item Update matching: \ $M \leftarrow M \cup (i, j)$. Define the weight of $(i, j)$ to be $b_j$.
			\end{enumerate} 
			\item {\bf Output:} Output matching $M$ and its total weight $W$. 
		\end{enumerate} 
				\end{enumerate} 
	% \bigskip
			\end{alg}
	\end{wbox}
\end{figure}

%\bigskip

\begin{definition}
	\label{def.R_j}
We will define several runs of Algorithm \ref{alg.main}. In these runs, we will assume Step (1) is executed once. We next define several ways of executing Step (2). Let $\cR$ denote the run of Step (2) on the entire graph $G$.  Corresponding to each bidder $j \in A$, let $G_j$ denote graph $G$ with bidder $j$ removed. Define $\cR_j$ to be the run of Step (2) on graph $G_j$. 
\end{definition}
%\bigskip

Lemma \ref{lem.S-subset-j} and Corollary \ref{cor.S-subset-j} given below establish a relationship between the available bidders for a query $i$ in the two runs $\cR$ and $\cR_j$. Note that bidders are available in multiplicity and therefore we will have to use the notion of a multiset rather than a set, as was done in \cite{Va.MFCS}.

A {\em multiset} contains elements with multiplicity. Let $A$ and $B$ be two multisets over $n$ elements $\{1, 2, \ldots n\}$, and let $a_i \geq 0$ and $b_i \geq 0$ denote the multiplicities of element $i$ in $A$ and $B$, respectively. We will say that $A \subseteq B$ if for each $i$, $a_i \leq b_i$, and $A = B$ if for each $i$, $a_i = b_i$. We will say that $i \in A$ if $a_i \geq 1$. We will define $A \cap B$ to be the multiset containing each element $i$ exactly $\min \{a_i, b_i \}$ times, and $A - B$ to be the multiset containing each element $i$ exactly $\max \{a_i - b_i, 0\}$ times. 

As before, let us renumber the queries so their order of arrival under $\cO(B)$ is $1, 2, \ldots n$. Let $T(i)$ and $T_j(i)$ denote the multisets of available bidders at the time of arrival of query $i$ (i.e., just before the query $i$ gets matched) in runs $\cR$ and $\cR_j$, respectively. In particular, $T(1)$ will contain $k_l$ copies of $l$ for each bidder $l$ and $T_j(1)$ will contain $k_l$ copies of $l$ for each bidder $l$, other than $j$. Similarly, let $S(i)$ and $S_j(i)$ denote the projections of $T(i)$ and $T_j(i)$ on the neighbors of $i$ in $G$ and $G_j$, respectively. 

We have assumed that Step (1) of Algorithm \ref{alg.main} has already been executed and a price $p_k$ has been assigned to each bidder $k$. The effective bid of bidder $k$ is $\eb(k) = b_k \cdot (1 - p_k)$. With probability 1, the effective bids of all bidders are distinct. Let $F_1$ be the multiset containing $k_l$ copies of $l$ for each $l \in A$ such that $b_l \cdot (1 - p_l) > b_j \cdot (1 - p_j)$. Similarly, let $F_2$ be the multiset containing $k_l$ copies of $l$ for each $l \in A$ such that and $b_l \cdot (1 - p_l) < b_j \cdot (1 - p_j)$. Observe that $j$ is not contained in either multiset.  

\bigskip

\begin{lemma}
	\label{lem.S-subset-j}
	For each $i$, $1 \leq i \leq n$, the following hold:
	\begin{enumerate}
		\item $(T_j(i) \cap F_1) =  (T(i) \cap F_1)$.
		\item $(T_j(i) \cap F_2) \subseteq  (T(i) \cap F_2)$.
	\end{enumerate}
\end{lemma}

\begin{proof}
	{\bf 1).} Clearly, in both runs, $\cR$ and $\cR_j$, any query having an available bidder in $F_1$ will match to the most profitable one of these, without even considering the rest of the bidders.  Since $j \notin F_1$, the two runs behave in an identical manner on the set $F_1$, thereby proving the first statement. 
	 
	{\bf 2).} The proof is by induction on $i$. The base case is trivially true because $(T_j(1) \cap F_2) =  (T(1) \cap F_2)$, since $j \notin F_2$. Assume that the statement is true for $i = k$ and let us prove it for $i = k+1$. By the first statement, we need to consider only the case that there are no available bidders for the $k^{th}$ query in $F_1$ in the runs $\cR$ and $\cR_j$. Assume that in run $\cR_j$, this query gets matched to bidder $l$; if it remains unmatched, we will take $l$ to be null. Clearly, $l$ is the most profitable bidder it is incident to in $T_j(k)$. Therefore, the most profitable bidder it is incident to in run $\cR$ is the best of $l$, the most profitable bidder in $T(k) - T_j(k)$, and $j$, in case it is available. In each of these cases, the induction step holds.
\end{proof}

In the corollary below, the first two statements follow from Lemma \ref{lem.S-subset-j} and the third statement follows from the first two statements. 

\bigskip

%The proof of this lemma is identical to that of Lemma \ref{lem.subset-j}, other than the use of multisets instead of sets, and is omitted. As in OBM, it uses the fact that the effective bid of a bidder $l$ is the same, namely $\eb(l) = b_l (1 - p_l)$, for any query which $l$ desires. 

\bigskip

\begin{corollary}
	\label{cor.S-subset-j}
		For each $i$, $1 \leq i \leq n$, the following hold:
	\begin{enumerate}
		\item $(S_j(i) \cap F_1) =  (S(i) \cap F_1)$.
		\item $(S_j(i) \cap F_2) \subseteq  (S(i) \cap F_2)$.
		\item $S_j(i) \subseteq S(i)$.
	\end{enumerate}
\end{corollary}

\bigskip

Next we define a new random variable, $u_e$, for each edge $e = (i, j) \in E$. This is called the {\em truncated threshold} for edge $e$ and is given in Definition \ref{def.threshold-Single}. It is critically used in the proofs of Lemmas \ref{lem.u_e-Single} and \ref{lem.star}.

\begin{definition}
	\label{def.threshold-Single}
	Let $e = (i, j) \in E$ be an arbitrary edge in $G$. Define random variable, $u_e$, called the {\em truncated threshold} for edge $e$, to be $u_e = \min \{\ut_i, \ b_j \cdot \left(1 - {\frac 1 e}\right) \}$, where $\ut_i$ is the utility of query $i$ in run $\cR_j$.
	\end{definition}
	
	\bigskip

\begin{definition}
	\label{def.j-star}
	Let $j \in A$. Henceforth, we will denote $k_j$ by $k$ in order to avoid triple subscripts. Let $i_1, \ldots , i_k$ be queries such that for $1 \leq l \leq k$, $(i_l, j) \in E$. Then $(j; \ i_1, \ldots , i_k)$ is called a {\em $j$-star}. Let $X_j$ denote this $j$-star. The contribution of $X_j$ to $\bE[W]$ is $\bE[r_j] + \sum_{l = 1}^k {\bE[u_{i_l}]}$, and it will be denote by $\bE [X_j]$.
\end{definition}	
	
	\bigskip

Corresponding to $j$-star $X_j = (j; \ i_1, \ldots , i_k)$, denote by $e_l$ the edge $(i_l, j) \in E$, for $1 \leq l \leq k$. Furthermore, let $u_{e_l}$ denote the truncated threshold random variable corresponding to $e_l$. 

\bigskip

\begin{property}
	\label{prop.SV-no-surpass}
	{\bf (No-Surpassing for \textsc{Single-Valued})} 
Assume that Step 1 of Algorithm \ref{alg.main} has been executed and a price $p_k$ has been assigned to each advertiser $k$. Suppose that the effective bid which query $i$ gets in run $\cR_j$ is less than $b_j \cdot (1 - p_j)$; the latter is clearly the effective bid which $j$ makes to $i$ in run $\cR$. Then, in run $\cR$, no bid to $i$ will surpass $\eb(j) = b_j \cdot (1 - p_j)$.
\end{property}

%The proof of the next lemma is similar to that of Lemma \ref{lem.OBM-no-surpassing}.

\bigskip

\begin{lemma}
	\label{lem.SV-no-surpassing}
	The No-Surpassing Property holds for \textsc{Single-Valued}. 
\end{lemma}

\begin{proof}
Suppose the bid of $j$, namely $b_j \cdot (1 - p_j)$, is better than the best bid that buyer $i$ gets in run $\cR_j$. If so, $i$ gets no bid from $F_1$ in $\cR_j$; observe that they are all higher than $b_j \cdot (1 - p_j)$. Now, by the first part of Corollary \ref{cor.S-subset-j}, $i$ gets no bid from $F_1$ in run $\cR$ as well, i.e., in run $\cR$, no bid to $i$ will surpass $b_j \cdot (1 - p_j)$.
\end{proof}

\bigskip

\begin{lemma}
	\label{lem.u_e-Single}
	Corresponding to $j$-star $X_j = (j; \ i_1, \ldots , i_k)$, the following hold. 
	\begin{itemize}
	    \item For $1 \leq l \leq k$, \ $u_{i_l} \geq u_{e_l}$. 
	   	\end{itemize}
\end{lemma}

\begin{proof}
By the third statement of Corollary \ref{cor.S-subset-j}, $i_l$ has more options in run $\cR$ as compared to run $\cR_j$. Furthermore, the truncation of the random variable only aids the inequality needed and therefore $u_{i_l} \geq u_{e_l}$.
\end{proof}

Our next goal is to lower bound the contribution of an arbitrary $j$-star, $\bE [X_j]$, which in turn involves lower bounding $\bE [r_j]$. The latter crucially uses the fact that $p_j$ is independent of $u_{e_l}$. This follows from the fact that $u_{e_l}$ is determined by run $\cR_j$ on graph $G_j$, which does not contain vertex $j$. 

\bigskip

\begin{lemma}
	\label{lem.star}
	Let $j \in A$ and let $X_j = (j; \ i_1, \ldots , i_k)$ be a $j$-star. Then 
	\[ \bE [X_j] \ \geq \ k \cdot b_j \cdot \left(1 - {\frac 1 e}\right) . \]
\end{lemma}

\begin{proof}	
	We will first lower bound $\bE [r_j]$. Let $f_U(b_j \cdot z_1, \ldots b_j \cdot z_k)$ be the joint probability density function of $(u_{e_1}, \ldots u_{e_k})$; clearly, $f_U(b_j \cdot z_1, \ldots b_j \cdot z_k)$ can be non-zero only if $z_l \in [0, 1 - {1 \over e}]$, for $1 \leq l \leq k$. By the law of total expectation, 
	
	\[ \bE [r_j] \ = \int_{(z_1, \ldots, z_k)} {\bE [r_j ~|~ u_{e_1} = b_j \cdot z_1, \ldots , u_{e_k} = b_j \cdot z_k] \cdot f_U(b_j \cdot z_1, \ldots b_j \cdot z_k) \ dz_1 \ldots dz_k} ,\]
	where the integral is over $z_l \in [0, \left(1 - {\frac 1 e}\right)]$, for $1 \leq l \leq k$. 
	
	For lower-bounding the conditional expectation in this integral, let $w_l \in [0, 1]$ be s.t. $e^{w_l -1} = 1-z_l$, for $1 \leq l \leq k$. For $x \in [0, 1]$, define the set $S(x) = \{l \ | \ 1 \leq l \leq k \ \mbox{and}\  x < w_l \}$. 
	
	%Next, define function $\alpha : \ [0, 1] \rightarrow \Zplus$ to be $|S(\alpha)|$.
		
	\begin{claim}
		\label{claim.alpha}
		Conditioned on $(u_{e_1} = b_j \cdot z_1, \ldots , u_{e_k} = b_j \cdot z_k)$, if $p_j = e^{x-1}$, then the degree of $j$ at the end of Algorithm $ \cA_2$ is at least $|S(x)|$, i.e., the contribution to $r_j$ in this run was $\geq b_j \cdot p_j \cdot |S(x)|$. 
	\end{claim}
	
\begin{proof}
Suppose $l \in S(x)$, then $x < w_l$. In run $\cR_j$, the maximum effective bid that $i_l$ received has value $b_j \cdot z_l$. In run $\cR$, if at the arrival of query $i_l$, $j$ is already fully matched, the contribution to $r_j$ in this run was $k \cdot b_j \cdot p_j$ and the claim is obviously true. If not, then since $x < w_l$, $b_j \cdot (1 - p_j) > b_j \cdot z_l$. The crux of the matter is that by Lemma \ref{lem.SV-no-surpassing}, the No-Surpassing Property holds. Therefore, query $i_l$ will receive its largest effective bid from $j$, $i_l$ will get matched to it, and $r_j$ will be incremented by $b_j \cdot p_j$. The claim follows.
\end{proof}

	For $1 \leq l \leq k$, define indicator functions $I_l : [0, 1] \rightarrow \{0, 1 \}$ as follows.
	   \[
            I_l(x) \ \ = \ \begin{cases}
                1 & \text{if } x < w_l, \\
                0 & \text{otherwise.}
            \end{cases}
        \]
Clearly, $|S(x)| = \sum_{l = 1}^k {I_j(x)}$. By Claim \ref{claim.alpha}, 
\[ \bE [r_j ~|~ u_{e_1} = b_j \cdot z_1, \ldots , u_{e_k} = b_j \cdot z_k]  \ \geq \ b_j \cdot \int_0^{1} {|S(x)| \cdot e^{x-1} \ dx}  \]

\[ = \ b_j \cdot \int_0^{1} {\sum_{l = 1}^k {I_l(x)} \cdot e^{x-1} \ dx} \
= \ b_j \cdot \sum_{l = 1}^k { \int_0^{1} { {I_l(x)} \cdot e^{x-1} } \ dx} \
  = \ b_j \cdot \sum_{l = 1}^k {\int_0^{w_l} {e^{x-1} \ dx} } \]
  
\[ = \ b_j \cdot \sum_{l = 1}^k {\left( e^{w_l -1} - {\frac 1 e} \right)} \ = \ b_j \cdot \sum_{l = 1}^k \left( 1 - {\frac 1 e} - z_l \right) . \]
	
Since $I_l(x) = 0$ for $x \in [w_l, 1]$, we get that $\int_0^{1} { {I_l(x)} \cdot e^{x-1} } \ dx \ = \ \int_0^{w_l} {e^{x-1} }\ dx$; this fact has been used above. Therefore, 
	\[ \bE [r_j] \ = \int_{(z_1, \ldots, z_k)} {\bE [r_j ~|~ u_{e_1} = b_j \cdot z_1, \ldots , u_{e_k} = b_j \cdot z_k] \cdot f_U(b_j \cdot z_1, \ldots b_j \cdot z_k) \ dz_1 \ldots dz_k} \]

\[ \geq b_j \cdot  \int_{(z_1, \ldots, z_k)}  {\sum_{l = 1}^k \left( 1 - {\frac 1 e} - z_l \right) \cdot f_U(b_j \cdot z_1, \ldots b_j \cdot z_k) \ dz_1 \ldots dz_k} \] 
	
\[	= \  k \cdot b_j \cdot \left( 1 - {\frac 1 e} \right) -  \sum_{l = 1}^k {\bE[u_{e_l}]}  ,\]
where both integrals are over $z_l \in [0, \left(1 - {\frac 1 e}\right)]$, for $1 \leq l \leq k$.

	By Lemma \ref{lem.u_e-Single}, $\bE [u_{i_l}] \geq \bE [u_{e_l}]$, for $1 \leq l \leq k$. Hence we get
\[ \bE [X_j]  \ = \ \bE[r_j] + \sum_{l = 1}^k {\bE[u_{i_l}]} \ 
\geq   \  k  \cdot b_j \cdot \left( 1 - {\frac 1 e} \right) , \]
\end{proof}

\bigskip

\begin{lemma}
	\label{lem.add}
	\[ \bE [W] \ =  \sum_i^n {\bE \left[ {u_i} \right]} \ + \   \sum_j^m {\bE[r_j]}  .\]
\end{lemma}

\begin{proof}
	By definition of the random variables, 
		\[ \bE [W] \ = \bE \left[ \sum_{i = 1}^n {u_i} \ + \ \sum_{j = 1}^m {r_j} \right] 
		\ =  \sum_i^n {\bE \left[ {u_i} \right]} \ + \   \sum_j^m {\bE[r_j]}   ,\]
where the first equality follows from the fact that if $(i, j) \in M$ then $W$ is incremented by $b_j$ and $u_i + r_j = b_j$. The second equality follows from linearity of expectation.  
\end{proof}

\begin{theorem}
	\label{thm.Single}
	The competitive ratio of Algorithm $ \cA_2$ is at least $1 - {\frac 1 e}$. Furthermore, it is budget-oblivious. 
\end{theorem}

\begin{proof}
	Let $P$ denote a maximum weight $b$-matching in $G$, computed in an offline manner. By the assumption made in Remark \ref{rem.complete-matching}, its weight is
	\[ w(P) = \sum_{j = 1}^m {k_j \cdot b_j} . \]
	 Let $T_j$ denote the $j$-star, under $P$, corresponding to each $j \in A$. The expected weight of matching produced by $ \cA_2$ is 
	\[ \bE \left[ W \right] \ = \ \sum_{i= 1}^n {\bE \left[ {u_i} \right]} \ + \   \sum_{j = 1}^m {\bE[r_j]}   	\ = \ \sum_{j=1}^m {\bE [T_j]} \ \geq \ \sum_{j = 1}^m {b_j \cdot k_j} \left( 1 - {\frac 1 e} \right) \ = \ \left( 1 - {\frac 1 e} \right) \cdot w(P) ,\]
	where the first equality uses Lemma \ref{lem.add}, the second follows from linearity of expectation and the inequality follows from Lemma \ref{lem.star}. 
	
	Finally, Algorithm $ \cA_2$ is budget-oblivious because it does not need to know $k_j$ for bidders $j$; it only needs to know during a run whether the $k_j$ bids available to bidder $j$ have been exhausted. The theorem follows. 
\end{proof}

%% file: Adwords.tex
\section{Algorithm for \textsc{Small}, After Assuming the No-Surpassing Property}
\label{sec.SMALL}

Two new difficulties arise for the problem \textsc{Adwords} The first is the inherent structural difficulty described in Section \ref{sec.diff-GENERAL}. Second, since bidders can have different bids for different queries, the no-surpassing property does not hold anymore, see Example \ref{ex.GEN-no-surpass}.

\begin{example}
	\label{ex.GEN-no-surpass}
	Assume that in the given instance for \textsc{Adwords}, $j, j'$ are two of the bidders, and $1, \ldots , k$ are $k$ of the queries, where $k$ is a large number. Assume $\bid(l, j) = \alpha$ for $1 \leq l \leq k$ and $\bid(l, j') = \alpha -1$ for $1 \leq l \leq k-1$. Further, assume that $\bid(k, j') = (\alpha - 1) \cdot (k-1)$. Let the budgets be $B_j = \alpha \cdot k$ and $B_{j'} = (\alpha - 1) \cdot (k-1)$. 
	
	Now consider a run in which $p_j = p_{j'} = p$. Assume that in run $\cR_j$\footnote{Run $\cR_j$ is defined in Definition \ref{def.R_j}.}, the best effective bid to $1, \ldots, k-1$ comes from $j'$, and in run $\cR_j$, the best effective bid to $1, \ldots, k-1$ comes from $j$. In run $\cR_j$, the budget of $j'$ is exhausted when $k$ arrives and assume that $k$ does not get any bids, making $u_e = 0$ for $e = (k, j)$. Now in run $\cR$, $\eb(k, j) = \alpha (1-p)$ and $\eb(k, j') = (\alpha - 1) \cdot (k-1) \cdot (1-p)$. Thus, even though $\eb(k, j) > u_e$, $k$ will be matched to $j'$ and not $j$. Clearly, this phenomenon will hold for all runs in which $p_{j'}$ is not too much larger than $p_j$.
\end{example}

\bigskip

\setcounter{figure}{1} 
%\begin{figure}[H]

\begin{figure}

	\begin{wbox}
		\begin{alg}
		\label{alg.adwords}
		{\bf ($\cA_2$:  Algorithm for \textsc{Adwords})}\\

		\begin{enumerate}
			\item {\bf Initialization:} $M \leftarrow \emptyset$, \ $W \leftarrow 0$ \ and \ $W_f \leftarrow 0$ \\
			 $\forall j \in A$, {\bf do}:
				\begin{enumerate}
				\item Pick $w_j$ uniformly from $[0, 1]$ and 
			set price \ $p_j \leftarrow e^{w_j - 1}$.
			\item $r_j \leftarrow 0$.
			\item $L_j \leftarrow B_j$.
			\end{enumerate}
			
		\bigskip
			\item {\bf Query arrival:} When query $i$ arrives, {\bf do}:
			\begin{enumerate}
			\item  $\forall j \in A$ s.t. $(i, j) \in E$ and $L_j > 0$ {\bf do}: 
			\begin{enumerate}
				\item $\eb(i, j) \leftarrow \bid(i, j) \cdot (1 - p_j)$.
				\item Offer effective bid of $\eb(i, j)$ to $i$.
			\end{enumerate}

			\item Query $i$ accepts the bidder whose effective bid is the largest. \\
			      (If there are no bids, matching $M$ remains unchanged.) \\
				If $i$ accepts $j$'s bid, then {\bf do}:
			\begin{enumerate}
			\item Set utility: \ $u_i \leftarrow \bid(i, j) \cdot (1 - p_j)$.
			\item Update revenue: \ $r_j \leftarrow r_j + \bid(i, j) \cdot p_j$.
			\item Update matching: \ $M \leftarrow M \cup (i, j)$. 
			\item Update weight: \ $W \leftarrow W + \min \{L_j, \bid(i, j)\}$ \ and \\ $W_f \leftarrow W_f + \max \{0, \ \bid(i, j) - L_j\}$.
			\item Update $L_j$: \ $L_j \leftarrow L_j - \min \{L_j, \bid(i, j)\}$.
			\end{enumerate}
			\end{enumerate} 
			\item {\bf Output:}	Output matching $M$, real money spent $W$, and fake money spent $W_f$. 
		\end{enumerate} 
			\end{alg}
	\end{wbox}
\end{figure}

For the rest of this section, we will make this assumption:

\bigskip

\noindent

{\bf Assumption of No-Surpassing for \textsc{Adwords}:} 
The following holds:\\
Assume that Step 1 of Algorithm \ref{alg.adwords} has been executed and a price $p_k$ has been assigned to each advertiser $k$. Suppose that the effective bid which query $i$ gets in run $\cR_j$ is less than $\bid(i, j) \cdot (1 - p_j)$; the latter is clearly the effective bid which $j$ makes to query $i$ in run $\cR$. Then, in run $\cR$, no bid to $i$ will surpass $\eb(i, j) = \bid(i, j) \cdot (1 - p_j)$.

\bigskip

Section \ref{sec.Adwords} presents an algorithm for \textsc{Adwords}, using fake money; the above-stated assumption is used in its analysis, in particular in the proof of Claim \ref{claim.alpha-gen}.  Section \ref{sec.Alg-Small} shows that by upper bounding the fake money used in the worst case, we get an optimal algorithm for \textsc{Small}, again based on the above-stated assumption.

\subsection{Structural Difficulties in \textsc{Adwords}}
 \label{sec.diff-GENERAL}
 
 To describe the structural difficulties in \textsc{Adwords}, we provide three instances in Example \ref{ex.three}. In order to obtain a completely unconditional result, we would need to adopt the following convention: assume bidder $j$ has $L_j$ money leftover and impression $i$ just arrived. Assume that $j$'s bid for $i$ is $\bid(i, j)$. If $\bid(i, j) > L_j$, then $j$ should not be allowed to bid for $i$, since $j$ has insufficient money. 
 
 Under this convention, it is easy to see that even a randomized algorithm will accrue only $\$W$ expected revenue on at least one of the instances given in Example \ref{ex.three}, provided it is greedy, i.e., if a match is possible, it does not rescind this possibility; the latter condition is a simple way of ensuring that the algorithm is ``fine tuned'' for a particular type of example. Note that the optimal for each instance is $\$2W$.

 \bigskip
  
 \begin{example}
 	\label{ex.three}
 	Let $W \in \Zplus$ be a large number. We define three instances of \textsc{Adwords}, each having two bidders, $b_1$ and $b_2$, with budgets of $\$W$ each. Instances $I_1$ and $I_2$ have $W+1$ queries, where for the first $W$ queries, both bidders bid $\$1$ each. For the last query, under $I_1$, $b_1$ bids $\$W$ and $b_2$ is not interested. Under $I_2$, $b_2$ bids $\$W$ and $b_1$ is not interested. Instance $I_3$ has $2W$ queries and both bidders bid $\$1$ for each of them.  
 \end{example}
 
Therefore, to obtain a non-trivial competitive ratio, bidder $j$ must be allowed to bid for $i$ even if $L_j < \bid(i, j)$. This amounts to the use of free disposal, since $j$ will be allowed to obtain query $i$ for less money than its value for $i$. Next, let's consider a second convention: if $L_j < \bid(i, j)$, then $j$ will bid $L_j$ for $i$. As stated in Remark \ref{rem.smaller-bid}, this convention is not supported by our proof technique, since Claim \ref{claim.alpha-gen} fails to hold, breaking the proof of Lemma \ref{lem.B-star-gen} and hence Lemma \ref{lem.Adwords}. 

This led us to a third convention: if $L_j < \bid(i, j)$, then $j$ will bid $L_j$ real money and $\bid(i, j) - L_j$ ``fake'' money for $i$. As a result, the total revenue of the algorithm consists of real money as well as fake money; in Algorithm \ref{alg.adwords}, these are denoted by $W$ and $W_f$, respectively. The problem now is that Lemma \ref{lem.Adwords}, which compares the total revenue of the algorithm, namely $W + W_f$, with the optimal offline revenue, does not yield the competitive ratio of Algorithm \ref{alg.adwords}. Remark \ref{rem.smaller-bid} explains why our proof technique does not allow us to dispense with the use of fake money. 

 We note that when Algorithm \ref{alg.adwords} is run on instances of \textsc{OBM}, it reduces to \textsc{Ranking}. Therefore, it is indeed a (simple) extension of \textsc{Ranking} to \textsc{Adwords}.

\subsection{Algorithm for \textsc{Adwords}}
\label{sec.Adwords}

Algorithm \ref{alg.adwords}, which will be denoted by $\cA_2$, is an attempt at online algorithm for \textsc{Adwords}. As stated in Section \ref{sec.diff-GENERAL}, because of the use of fake money, we will not be able to give a competitive ratio for it, instead, in Lemma \ref{lem.Adwords}, we will compare the sum of real and fake money spent by the algorithm with the real money spent by an  optimal offline algorithm. 

In algorithm $\cA_2$, $L_j \in \Zplus$ will denote bidder $j$'s {\em leftover budget}; it is initialized to $B_j$. At the arrival of query $i$, bidder $j$ will bid for $i$ if $(i, j) \in E$ and $L_j > 0$. In general, $i$ will receive a number of bids. The exact procedure used by $i$ to accept one of these bids is given in algorithm $\cA_2$; its steps are self-explanatory. If $i$ accepts $j$'s bid then $i$ is matched to $j$, the edge $(i, j)$ is assigned a weight of $\bid(i, j)$ and $L_j$ is decremented by $\min \{L_j, \bid(i, j) \}$.

Note that we do not require that there is sufficient left-over money, i.e., $L_j \geq \bid(i, j)$,  for $j$ to bid for $i$. In case $L_j < \bid(i, j)$ and $i$ accepts $j$'s bid, then $\bid(i, j) - L_j$ of the money paid by $j$ for $i$ is {\em fake money}; this will be accounted for by incrementing $W_f$ by $\bid(i, j) - L_j$. The rest, namely $L_j$, is real money and is added to $W$. If $\bid(i, j) \geq  L_j$ and $i$ accepts $j$'s bid, then $L_j$ becomes zero and $j$ does not bid for any future queries. At the end of the algorithm, random variable $W$ denotes the total real money spent and $W_f$ denotes the total fake money spent. 

The {\em offline optimal solution} to this problem is defined to be a matching of queries to advertisers that maximizes the weight of the matching; this is done with full knowledge of graph $G$. As stated in Remark \ref{rem.complete-matching}, we will assume that under such a matching, $P$, the budget $B_j$ of each bidder $j$ is fully spent, i.e., $w(P) = \sum_{j = 1}^m {B_j}$.

%%%%%%%%%%%%%%%%%%%%%%%%%%%%%%%%%%%%%%%%%%%%%%%%%%%%%%%%%%%

\begin{figure}[h]
\begin{center}
\includegraphics[width=7.2in]{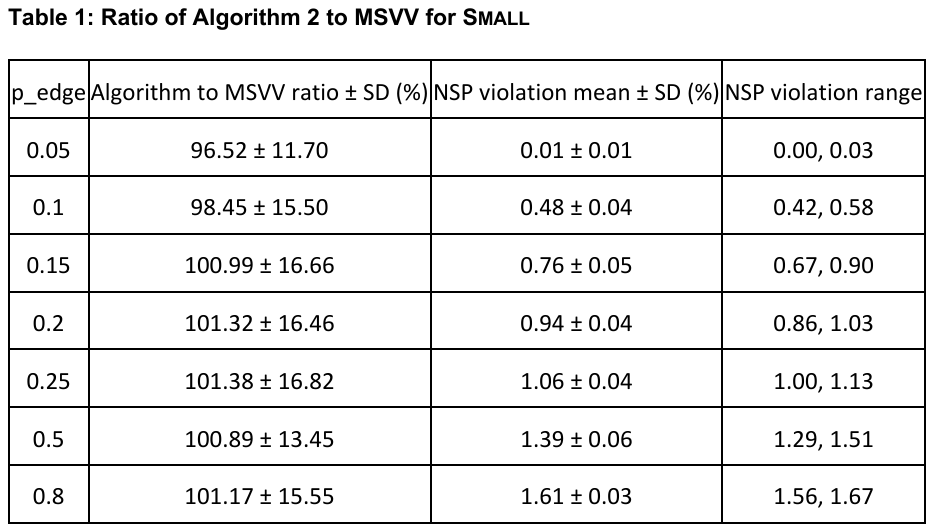}
%\caption{}
\label{fig.1}
\end{center}
\end{figure}

%%%%%%%%%%%%%%%%%%%%%%%%%%%%%%%%%%%%%%%%%%%%%%%%%%%%%%%%%%%%%%%

%%%%%%%%%%%%%%%%%%%%%%%%%%%%%%%%%%%%%%%%%%%%%%%%%%%%%%%%%%%

\begin{figure}[h]
\begin{center}
\includegraphics[width=7.2in]{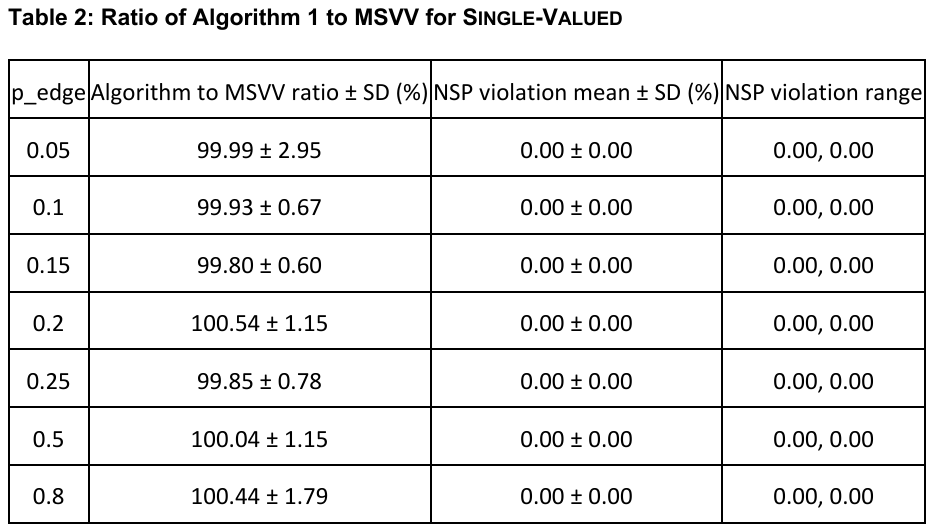}
%\caption{}
\label{fig.2}
\end{center}
\end{figure}

%%%%%%%%%%%%%%%%%%%%%%%%%%%%%%%%%%%%%%%%%%%%%%%%%%%%%%%%%%%%%%%

%%%%%%%%%%%%%%%%%%%%%%%%%%%%%%%%%%%%%%%%%%%%%%%%%%%%%%%%%%%%%%%
\begin{figure}[ht]
\begin{minipage}[b]{0.5\linewidth}
\centering
\includegraphics[width=7in]{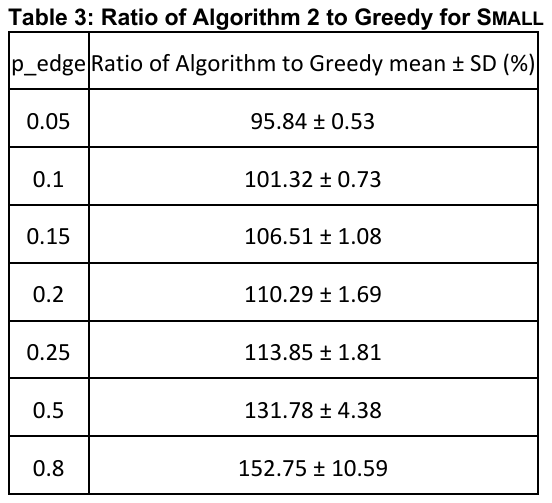}
%\caption{Can DDFS be performed on bridge $(u, v)$ at search level 6?}
\label{fig.3}
\end{minipage}
\hspace{0.5cm}
\begin{minipage}[b]{0.5\linewidth}
\centering
\includegraphics[width=7in]{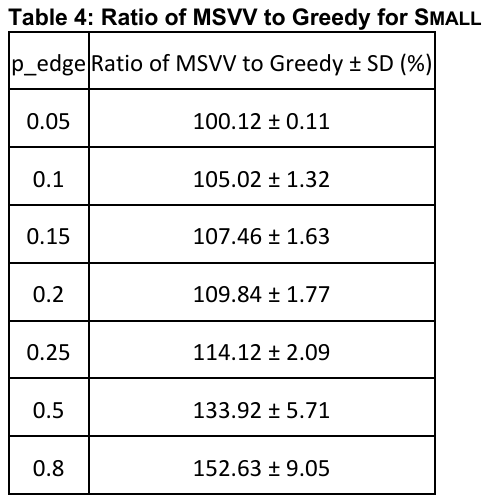}
%\caption{If so, vertices $a$ and $b$ will get wrong tenacities.}
\label{fig.4}
\end{minipage}
\end{figure}
%%%%%%%%%%%%%%%%%%%%%%%%%%%%%%%%%%%%%%%%%%%%%%%%%%%%%%%%%%%%%%%

\subsection{Analysis of Algorithm \ref{alg.adwords}}
\label{sec.analysis-gen}

\begin{lemma}
	\label{lem.add-gen}
	\[ \bE [W + W_f] \ =  \sum_i^n {\bE \left[ {u_i} \right]} \ + \   \sum_j^m {\bE[r_j]}  .\]
\end{lemma}

\begin{proof}
	For each edge $(i, j) \in M$, its contribution to $W + W_f$ is $\bid(i, j)$. Furthermore, the sum of $u_i$ and the contribution of $(i, j)$ to $r_j$ is also $\bid(i, j)$. This gives the first equality below. The second equality follows from linearity of expectation.
		\[ \bE [W + W_f] \ = \bE \left[ \sum_{i = 1}^n {u_i} \ + \ \sum_{j = 1}^m {r_j} \right] 
		\ =  \sum_i^n {\bE \left[ {u_i} \right]} \ + \   \sum_j^m {\bE[r_j]}   ,\]
\end{proof}

Recall that for \textsc{Single-Valued}, we gave Lemma \ref{lem.S-subset-j} and Corollary \ref{cor.S-subset-j}, which established a relationship between the available bidders for a query $i$ in the two runs $\cR$ and $\cR_j$. These facts dealt with multisets and in Section \ref{sec.Single-Adwords}, we have defined operations on multisets. 

We will need Lemma \ref{lem.S-subset-j} and Corollary \ref{cor.S-subset-j} for analyzing Algorithm \ref{alg.adwords} as well, though the definitions of the multisets will be guided by the following: If bidder $k \in A$ has leftover money of $L_k$, as determined by Algorithm \ref{alg.adwords}, then we will say that $i$ has $L_k$ copies of $k$ {\em available} to it. Furthermore, if $k$'s bid for $i$ is $\bid(i, k)$ and this bid is successful, then $L_k$ will be decremented by $\min \{L_k, \bid(i, k)\}$, as stated in Step 2(b)(v) of the algorithm, and the available copies of $k$ for the next bidder will decrease accordingly. 

As before, let us renumber the queries so their order of arrival under $\cO(B)$ is $1, 2, \ldots n$. Let $T(i)$ and $T_j(i)$ denote the multisets of available copies of each bidder at the time of arrival of query $i$ (i.e., just before the query $i$ gets matched), in runs $\cR$ and $\cR_j$, respectively. Similarly, let $S(i)$ and $S_j(i)$ denote the multisets obtained by restricting $T(i)$ and $T_j(i)$ to the bidders that have edges to query $i$ in graphs $G$ and $G_j$,  respectively.

We have assumed that Step (1) of Algorithm \ref{alg.adwords} has already been executed and a price $p_k$ has been assigned to each good $k$. With probability 1, the prices are all distinct. Let $F_1$ be the multiset containing $B_l$ copies of $l$ for each $l \in A$ such that $p_l < p_j$. Similarly, let $F_2$ be the multiset containing $B_l$ copies of $l$ for each $l \in A$ such that and $p_l > p_j$. 

Under the definitions and operations stated above, it is easy to check that Lemma \ref{lem.S-subset-j} and Corollary \ref{cor.S-subset-j} hold for Algorithm \ref{alg.adwords} as well. Therefore, Lemma \ref{lem.u_e-Single} also carries over. Definition \ref{def.threshold-Single} needs to be modified to the following.

\bigskip

\begin{definition}
	\label{def.threshold-Adwords}
	Let $e = (i, j) \in E$ be an arbitrary edge in $G$. Define random variable, $u_e$, called the {\em truncated threshold} for edge $e$, to be $u_e = \min \{u_i, \ \bid(i, j) \cdot \left(1 - {\frac 1 e}\right) \}$, where $u_i$ is the utility of query $i$ in run $\cR_j$.
	\end{definition}

\bigskip

Definition \ref{def.j-star} needs to be changed to the following. 

\bigskip

\begin{definition}
	\label{def.j-star-gen}
	Let $j \in A$. Let $i_1, \ldots , i_k$ be queries such that for $1 \leq l \leq k$, $(i_l, j) \in E$ and $\sum_{l = 1}^{k} {\bid(i_l, j)} = B_i$. Then $(j; \ i_1, \ldots , i_k)$ is called a {\em $B_j$-star}. Let $X_j$ denote this $B_j$-star. The contribution of $X_j$ to $\bE[W]$ is $\bE[r_j] + \sum_{l = 1}^k {\bE[u_{i_l}]}$, and it will be denote by $\bE [X_j]$.
\end{definition}	
	
Corresponding to $B_j$-star $X_j = (j; \ i_1, \ldots , i_k)$, denote by $e_l$ the edge $(i_l, j) \in E$, for $1 \leq l \leq k$. Furthermore, let $u_{e_l}$ denote the truncated threshold random variable corresponding to $e_l$. The next lemma crucially uses the fact that $p_j$ is independent of $u_{e_l}$; the reason for this fact is the same as in \textsc{Single-Valued}.

\bigskip

\begin{lemma}
	\label{lem.B-star-gen}
	Let $j \in A$ and let $X_j = (j; \ i_1, \ldots , i_k)$ be a $B_j$-star. Then 
	\[ \bE [X_j] \ \geq \  B_j \cdot \left(1 - {\frac 1 e}\right) . \]
\end{lemma}

\begin{proof}
We will first lower bound $\bE [r_j]$. Let $f_U(\bid(i_1, j) \cdot z_1, \ \ldots \ , \bid(i_k, j) \cdot z_k)$ be the joint probability density function of $(u_{e_1}, \ldots u_{e_k})$; clearly, $f_U(\bid(i_1, j) \cdot z_1, \ \ldots \ ,  \bid(i_k, j) \cdot z_k)$ can be non-zero only if $z_l \in [0, 1 - {1 \over e}]$, for $1 \leq l \leq k$. 

By the law of total expectation, \ \ \ \  $\bE [r_j] =$
	\[ \int_{(z_1, \ldots, z_k)} {\bE [r_j ~|~ u_{e_1} = \bid(i_1, j) \cdot z_1, \ldots , u_{e_k} = \bid(i_k, j) \cdot z_k] \cdot f_U(\bid(i_1, j) \cdot z_1, \ldots \bid(i_k, j) \cdot z_k) \ dz_1 \ldots dz_k} ,\]
	where the integral is over $z_l \in [0, \left(1 - {\frac 1 e}\right)]$, for $1 \leq l \leq k$. 
	
	For lower-bounding the conditional expectation in this integral, let $w_l \in [0, 1]$ be s.t. $e^{w_l -1} = 1-z_l$, for $1 \leq l \leq k$. Let $x \in [0, 1]$. For $1 \leq l \leq k$, define indicator functions $I_l : [0, 1] \rightarrow \{0, 1 \}$ as follows.
	   \[
            I_l(x) \ \ = \ \begin{cases}
                1 & \text{if } x < w_l, \\
                0 & \text{otherwise.}
            \end{cases}
        \]

	 Furthermore, define
	\[ V(x) = \sum_{l =1}^k { I_l(x) \cdot \bid(i_l, j) } . \]

	\begin{claim}
		\label{claim.alpha-gen}
		Conditioned on $(u_{e_1} = \bid(i_1, j) \cdot z_1, \ldots , u_{e_k} = \bid(i_k, j) \cdot z_k)$, if $p_j = e^{x-1}$, where $x \in [0, 1]$, then the contribution to $r_j$ in this run of algorithm $\cA_2$ was $\geq p_j \cdot V(x)$. 
	\end{claim}
	
	\begin{proof}
			Suppose $I_l(x) = 1$, then $x < w_l$. In run $\cR_j$, the maximum effective bid that $i_l$ received has value $\bid(i_l, j) \cdot z_l$. In run $\cR$, if on the arrival of query $i_l$, $L_j = 0$, i.e., $j$ is already fully matched, then the contribution to $r_j$ in this run was $B_j \cdot p_j$ and the claim is obviously true. If $L_j > 0$, then since $x < w_l$, $1 - p_j > z_l$. Therefore, by Corollary \ref{cor.S-subset-j}, query $i_l$ will receive its largest effective bid from $j$. Hence, $i_l$ will get matched to $j$ and $r_j$ will be incremented by $\bid(i_l, j) \cdot p_j$. The claim follows.
	\end{proof}

By Claim \ref{claim.alpha-gen}, 
\[ \bE [r_j ~|~ u_{e_1} = \bid(i_1, j) \cdot z_1, \ldots , u_{e_k} = \bid(i_k, j) \cdot z_k] \ \geq \ \int_0^{1} {V(x) \cdot e^{x-1} \ dx} \]

\[ = \sum_{l = 1}^k {\bid(i_l, j) \cdot {\int_0^{1} {I_l(x) \cdot e^{x-1} \ dx} } } 
\  = \  \sum_{l = 1}^k {\bid(i_l, j) \cdot {\int_0^{w_l} {e^{x-1} \ dx} } } \]

\[ = \ B_j \cdot \sum_{l = 1}^k {\left( e^{w_l -1} - {\frac 1 e} \right)} \ = \ B_j \cdot \sum_{l = 1}^k \left( 1 - {\frac 1 e} - z_l \right) . \]
	
Therefore, 	$ \bE [r_j] \ = $
	\[ \int_{(z_1, \ldots, z_k)} {\bE [r_j ~|~ u_{e_1} = \bid(i_1, j) \cdot z_1, \ldots , u_{e_k} = \bid(i_k, j) \cdot z_k] \cdot f_U(\bid(i_1, j) \cdot z_1, \ldots \bid(i_k, j) \cdot z_k) \ dz_1 \ldots dz_k} ,\]

\[ \geq B_j \cdot  \int_{(z_1, \ldots, z_k)}  {\sum_{l = 1}^k \left( 1 - {\frac 1 e} - z_l \right) \cdot f_U(\bid(i_1, j) \cdot z_1, \ldots \bid(i_k, j) \cdot z_k) \ dz_1 \ldots dz_k} \]
	
\[	= \  B_j \cdot \left( 1 - {\frac 1 e} \right) -  \sum_{l = 1}^k {\bE[u_{e_l}]}  .\]

	By Lemma \ref{lem.u_e-Single}, $\bE [u_{i_l}] \geq \bE [u_{e_l}]$, for $1 \leq l \leq k$. Hence we get

\[ \bE [X_j]  \ = \ \bE[r_j] + \sum_{l = 1}^k {\bE[u_{i_l}]} \ 
\geq   \ =  B_j \cdot \left( 1 - {\frac 1 e} \right) , \]
\end{proof}

\begin{lemma}
		\label{lem.Adwords}
	 Algorithm $\cA_2$ satisfies
	 \[ \bE \left[ W + W_f \right] \ \geq \ \left( 1 - {\frac 1 e} \right) \cdot w(P) .\]
	 Furthermore, it is budget-oblivious.
\end{lemma}

\begin{proof}
	Let $P$ denote a maximum weight $b$-matching in $G$. By the assumption made in Remark \ref{rem.complete-matching}, its weight is
	\[ w(P) = \sum_{j = 1}^m {B_j} . \]
	 Let $T_j$ denote the $j$-star, under $P$, corresponding to each $j \in A$. The expected weight of matching produced by $\cA_2$ is 
	\[ \bE \left[ W + W_f \right] \ = \ \sum_{i= 1}^n {\bE \left[ {u_i} \right]} \ + \   \sum_{j = 1}^m {\bE[r_j]}   	\ = \ \sum_{j=1}^m {\bE [T_j]} \ \geq \ \sum_{j = 1}^m {B_j} \cdot \left( 1 - {\frac 1 e} \right) \ = \ \left( 1 - {\frac 1 e} \right) \cdot w(P) ,\]
	where the first equality uses Lemma \ref{lem.add}, the second follows from linearity of expectation and the inequality follows by using Lemma \ref{lem.B-star-gen}. 
	
		Finally, Algorithm $ \cA_3$ is budget-oblivious because it does not need to know the budgets $B_j$ for bidders $j$; it only needs to know during a run whether $B_j$ has been exhausted. The lemma follows. 
\end{proof}

\bigskip

\begin{remark}
	\label{rem.smaller-bid}
	Let us consider the following two avenues for dispensing with the use of fake money altogether; we will show places where our proof technique breaks down for each one. Assume $L_j < \bid(i, j)$.
	\begin{enumerate}
		\item Why not modify Step 2 of Algorithm \ref{alg.adwords} so that $j$'s bid for $i$ is taken to be $L_j$ instead of $\bid(i, j)$?
		\item Why not modify Step 2(b)(i) so it sets $u_i$ to $L_j \cdot (1 - p_j)$ rather than $B_j \cdot (1 - p_j)$
			\end{enumerate}
				Under the first avenue, we cannot ensure $u_i \geq u_e$, since it may happen that    $u_e > L_j \cdot (1 - p_j) = u_i$. The condition $u_i \geq u_e$ is used for deriving $\bE [u_i] \geq \bE [u_e]$, which is essential in the proof of Lemma \ref{lem.B-star-gen}.  

		To make the second avenue work, the proof of Claim \ref{claim.alpha-gen} would need to be changed as follows: the last case, $L_j > 0$, will need to be split into the two cases given above. However, under Case 2, which applies if $L_j < \bid(i, j)$, even though $p_j < p$, the largest effective bid that query $i_l$ receives may not be the one from $j$, since the effective bid of $j$ has value $L_j \cdot (1 - p_j) < \bid(i_l, j) \cdot (1 - p_j)$. Therefore, $i_l$ may not get matched to $j$, thereby invalidating Claim \ref{claim.alpha-gen}. 
\end{remark}		

\subsection{Algorithm for \textsc{Small}}
\label{sec.Alg-Small}

We will use Lemma \ref{lem.Adwords} to show that Algorithm \ref{alg.adwords} yields algorithms for \textsc{Small} by upper bounding the fake money used in the worst case. Their budget-obliviousness follows from that of Algorithm \ref{alg.adwords}.

\bigskip

\begin{ctheorem}
	\label{thm.Adwords-small}
		 Algorithm $\cA_2$ is an optimal online algorithm for \textsc{Small}; furthermore, it is budget-oblivious.
\end{ctheorem}

\begin{proof}
Let $I$ be an instance of \textsc{Small}. 
\[W_f \leq \  {\sum_{j \in A} {\max_{(i, j) \in E} \  \{\bid(i, j) - 1\} } } \]
Therefore, 
\[ \mu(I) = \max_{j \in A} \ \ \left\{ {{\max_{(i, j) \in E} \ \{\bid(i, j) -1 \}} \over {B_j}} \right\} \geq 
{ { \sum_{j \in A}  {\max_{(i, j) \in E} \ \{\bid(i, j) -1 \}} } \over {\sum_{j \in A} {B_j}} } 
\geq {{W_f} \over {w(P)}} , \]

where $\mu(I)$ is defined in Section \ref{sec.prelim}. 
Now, by definition of \textsc{Small},
\[ \lim_{n(I) \rightarrow \infty} \ {\mu(I)} = 0 , \]
where $n(I)$ denotes the number of queries in instance $I$. 

Therefore 
\[ \lim_{n(I) \rightarrow \infty} \ {{W_f} \over {w(P)}} = 0 . \]
The theorem follows from Lemma \ref{lem.Adwords}.
\end{proof}

\subsection{Experimental Results}
\label{sec.expt}

The purpose of the experimental results is two-fold: first, to determine how often is the No-Surpassing Property violated on a typical instance and second, to evaluate the performance of Algorithm \ref{alg.adwords}. The results are summarized in the four Tables given. The instances were generated using the following parameters: 
\begin{itemize}
	\item The number of advertisers and queries was 20 and 2,000, respectively.
	\item Budgets were picked in the range $[100, 2000]$.
	\item Bids were picked in the range $[1, 20]$. In addition, for each advertiser, we ensured that the bids were at most $0.02$ times the budget, to ensure that bids were small compared to budgets.   
\end{itemize}

The instances of \textsc{Small} were generated as follows. We picked the edge densities of the underlying graph to be $0.05, \ \ 0.1, \ \ 0.15, \ \ 0.2, \ \ 0.25, \ \ 0.5$ and \ $0.8$. For each edge density, 20 underlying graphs were constructed at random. Then, budgets were assigned randomly to each advertiser, in the range specified. Finally, for each advertiser, bids were assigned to incident edges randomly, in the range specified. The order of arrival of queries was picked at random. This resulted in one instance of \textsc{Small}. 

For generating instances of \textsc{Single-Valued}, the underlying graphs were constructed as above. An instance of \textsc{Single-Valued} was obtained from such a graph as follows. For each advertiser $j$, the bid value $b_j$ was picked randomly from the integral interval $[1, 20]$. Then a random number $n_j$ was picked randomly from the interval $[100, 2000]$ and $k_j$ was set to $\lfloor{n_j \over b_j}\rfloor$. 

For each instance of \textsc{Small}, Algorithm \ref{alg.adwords} was run 40 times, each time randomly setting $p_j$ for each advertiser $j$. The average revenue of these runs was computed and was divided by the revenue of the MSVV Algorithm run on the same instance. Finally, for each edge density, the average of these ratios over all 20 instances was computed and noted in Table 1. A similar procedure was followed for instances of \textsc{Single-Valued} using Algorithm \ref{alg.main} and the MSVV Algorithm. These results are reported in Table 2. 

Tables 3 compares, for all edge densities and instances of \textsc{Small} constructed above,  Algorithm \ref{alg.adwords} to the greedy algorithm, under which each query is matched to the highest bid. Finally, Tables 4 compares the MSVV Algorithm to the greedy algorithm. Observe that the performances of Algorithm \ref{alg.adwords} and the MSVV Algorithm are far superior to that of the greedy algorithm, thereby indicating that these instances are difficult enough to not succumb to a trivial algorithm and lending more credence to the results reported in Tables 1 and 2.

Violations of the No-Surpassing Property were computed as follows. For each instance of \textsc{Small} and for each setting of random prices $p_j$ for advertisers $j$, we determined if a violation occurred for each edge $(i, j)$ of the underlying graph, where $i$ is a query and $j$ is an advertiser, as follows: check if the effective bid made to query $i$ in run $\cR_j$ is less than $\bid(i, j) \cdot (1 - p_j)$; the latter being the effective bid which $j$ makes to query $i$ in run $\cR$. If so, then if in run $\cR$, the effective bid to $i$ surpasses $\eb(i_l, j) = \bid(i, j) \cdot (1 - p_j)$, then we say that a violation of the No-Surpassing property has occurred for edge $(i, j)$. In Table 1, we have reported the percentage of edges for which violations happen. 

As shown in Lemma \ref{lem.SV-no-surpassing}, the No-Surpassing Property holds for \textsc{Single-Valued}. As a stress test of our code, we repeated the above experiment on instances of \textsc{Single-Valued} as well. The results are reported in Table 2; all entries are zero, as expected.

%% file: discussion.tex
\section{Discussion}
\label{sec.discussion}

The open question mentioned in the Introduction, of removing the assumption of no-surpassing property from our proof of Algorithm \ref{alg.adwords} for \textsc{Small}, deserves special attention because of its potential impact in the ad auctions marketplace. Another question is to place an upper bound on the expected fake money used, $\bE \left[W_f \right]$, in Algorithm \ref{alg.adwords} and strengthen Lemma \ref{lem.Adwords} to obtain a good bound on the competitive ratio of this algorithm for \textsc{Adwords}. This seems a promising avenue for improving the bound for \textsc{Adwords} from $0.5016$, given in \cite{Huang2020adwords}.

%% file: ack.tex
\section{Acknowledgements}
\label{sec.ack}

I wish to thank Asaf Ferber, Alon Orlitsky, Will Ma, Thorben Tr{\"o}bst and Rajan Udwani for valuable comments and discussions. A big thanks to Jacob Tyler Bigham and Joseph Jin-Him Wong, undergrads at UCI, for conducting the extensive experiments reported in this paper.